\documentclass[titlepage]{vgtc}
\let\ifpdf\relax

\usepackage[utf8]{inputenc}

\usepackage{cite}

\usepackage{subfig}
\usepackage{xcolor}

\usepackage{xspace}
\usepackage{enumerate}

\usepackage{microtype}

\usepackage{amsthm}
\usepackage{amsmath}
\usepackage{amssymb}
\usepackage{amstext}
\PassOptionsToPackage{hyphens}{url}
\usepackage[breaklinks]{hyperref}

\usepackage{mathptmx}
\usepackage{graphicx}
\usepackage{times}

\onlineid{0}
\vgtccategory{Research}
\vgtcinsertpkg

\usepackage{ifpdf}

\makeatletter
\def\url@leostyle{%
  \@ifundefined{selectfont}{\def\UrlFont{\sf}}{\def\UrlFont{\small\ttfamily}}}
\makeatother
\def\comment#1{}%
\def\withcomments{%
  \newcounter{mycommentcounter}%
   \def\comment##1{\refstepcounter{mycommentcounter}%
    \ifhmode%
     \unskip%
     {\dimen1=\baselineskip \divide\dimen1 by 2 %
       \raise\dimen1\llap{\tiny
	{-\themycommentcounter-}}}\fi%
     \marginpar[{\renewcommand{\baselinestretch}{0.8}%
       \hspace*{-2em}\begin{minipage}{1.2\marginparwidth}\footnotesize%
[\themycommentcounter]:%
\raggedright ##1\end{minipage}}]{\renewcommand{\baselinestretch}{0.8}%
       \begin{minipage}{1.2\marginparwidth}\footnotesize%
[\themycommentcounter]: \raggedright%
##1\end{minipage}}}%
  }

\newcommand{\CAC}{\textsc{Circular-Arc Cartogram}\xspace}

\DeclareCaptionType{copyrightbox}

\title{Circular-Arc Cartograms}

\author{
Jan-Hinrich Kämper\thanks{e-mail: kaemper@wi-ol.de}\\
\scriptsize Universität Oldenburg, Germany
\and Stephen G.~Kobourov\thanks{e-mail: kobourov@cs.arizona.edu}\\
\scriptsize University of Arizona, Tucson, Arizona, USA
\and Martin Nöllenburg\thanks{e-mail: noellenburg@kit.edu}\\
\scriptsize Karlsruhe Institute of Technology, Germany
}

\newtheorem{theorem}{Theorem}

\newtheorem{problem}{Problem}

\abstract{

    We present a new {\em circular-arc cartogram} model in which countries are
    drawn as polygons with circular arcs instead of straight-line segments.
    Given a political map and values associated with each country in the map,
    a cartogram is a distorted map in which the areas of the countries are
    proportional to the corresponding values. In the circular-arc cartogram
    model straight-line segments can be replaced by circular arcs in order to
    modify the areas of the polygons, while the corners of the polygons remain
    fixed. The countries in circular-arc cartograms have the 
aesthetically pleasing 
		appearance of clouds or snowflakes, depending on whether their
    edges are bent outwards or inwards. This makes it easy to determine
    whether a country has grown or shrunk, just by its overall shape. We show
    that determining whether a given map and given area-values can be realized
    as a circular-arc cartogram is an NP-hard problem. Next we describe a heuristic method for constructing circular-arc cartograms, which
    uses a max-flow computation on the dual graph of the map, along with a
    computation of the straight skeleton of the underlying polygonal
    decomposition. Our method is implemented and produces cartograms that,
    while not yet perfectly accurate, achieve many of the desired areas in our
    real-world examples.

}

\CCScatlist{ 
  \CCScat{I.3.5}%
{Computer Graphics}{Computational Geometry and Object Modeling}{}
}

\teaser{
\centering
\centering
		\subfloat[Geographically accurate map\label{fig:uselection-geo}]{\makebox[.45\textwidth]{\includegraphics[height=.16\textheight]{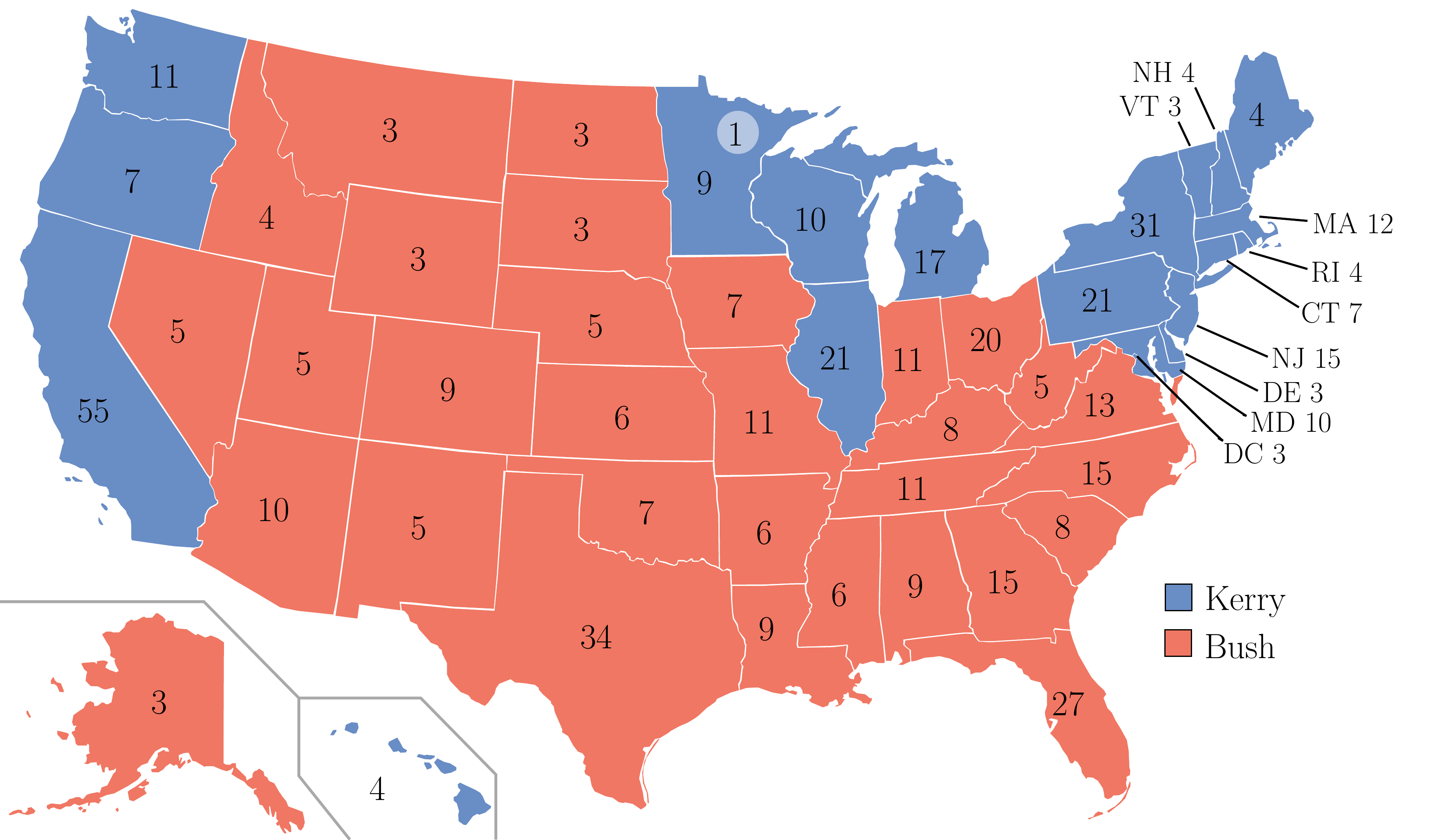}}}
		\hfill
	\subfloat[Rectilinear value-by-area cartogram by ChrisnHouston~\cite{rectielec04} \label{fig:uselection-recti}]{\makebox[.5\textwidth]{\includegraphics[height=.16\textheight]{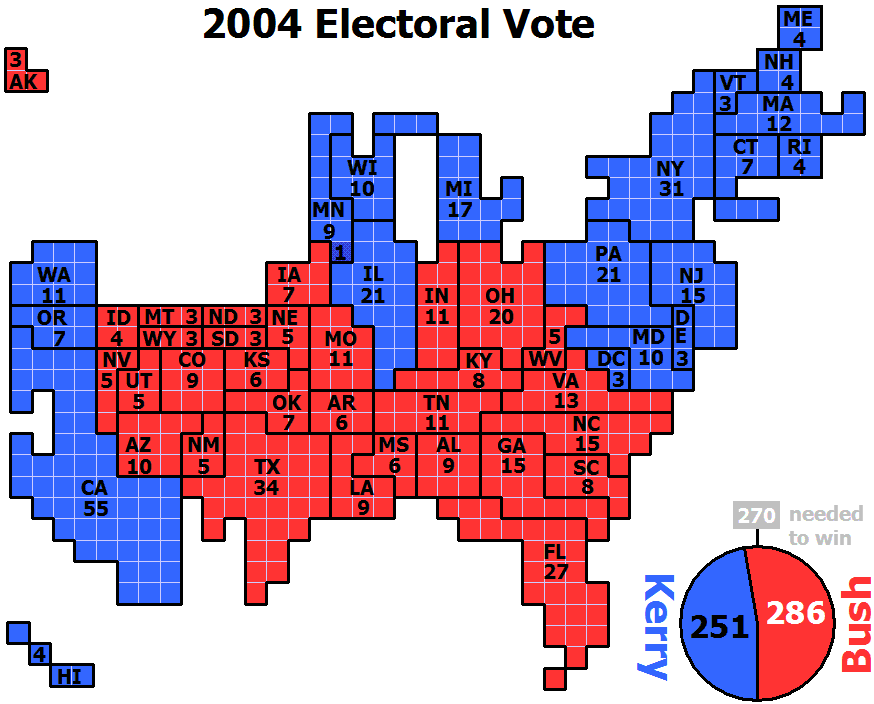}}}	

     \hfill
        \subfloat[Value-by-area cartogram generated by the diffusion method~\cite{Gastner}, used by permission of M.~Newman~\cite{newmanelec04}\label{fig:uselection-newman}]{\makebox[.48\textwidth]{\includegraphics[height=.16\textheight]{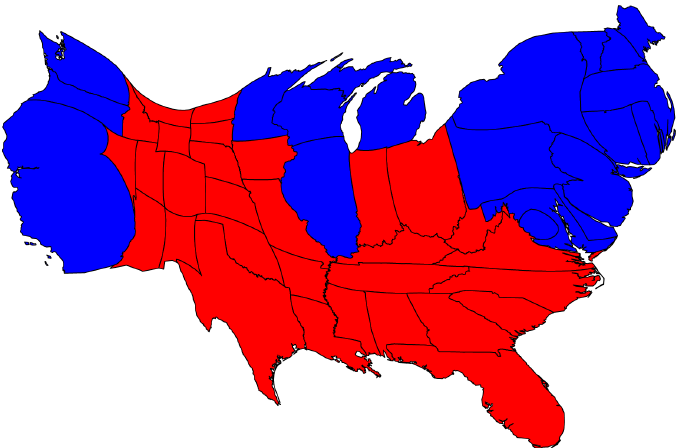}}}	
       \hfill
	\subfloat[Circular-arc cartogram generated by our method\label{fig:uselection-circ}]{\makebox[.48\textwidth]{\includegraphics[height=.16\textheight]{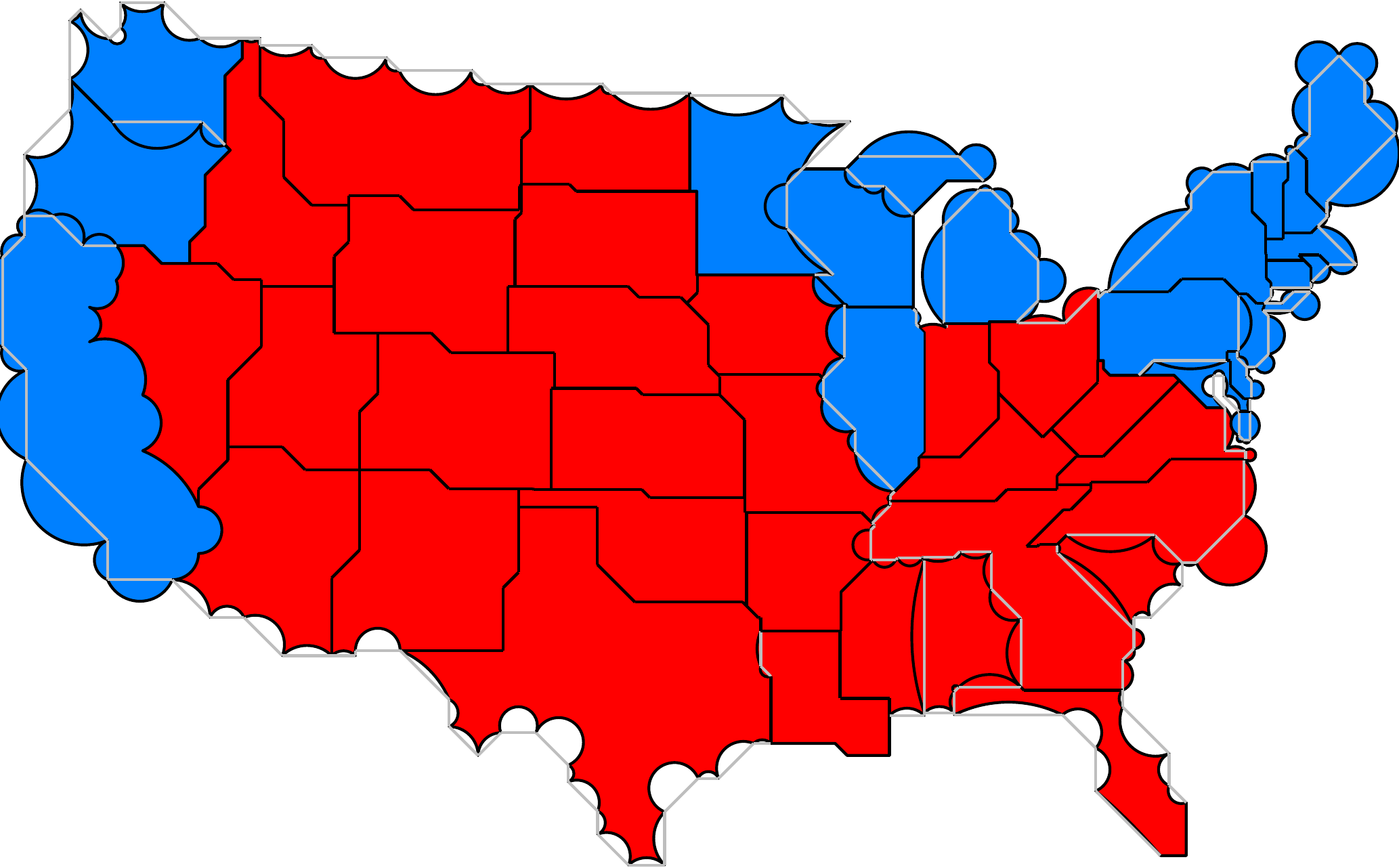}}}
	\caption{The results of the 2004 US presidential election. 
	In red states the majority of the vote was for the republican (G.~W.~Bush) and in blue states the majority was for the democrat (J.~Kerry). Note that the circular-arc cartogram makes it easy to see that most blue states are densly populated (cloud shapes), while the red states are more rural. Exceptions such as Oregon (blue but not densly populated) and North Carolina (red but dense) stand out.}
	\label{fig:usa}
}

\begin{document}

\ifpdf
\DeclareGraphicsExtensions{.pdf, .jpg, .tif}
\else
\DeclareGraphicsExtensions{.eps, .jpg}
\fi

\firstsection{Introduction}
\maketitle

	A \emph{cartogram}, or \emph{value-by-area diagram}, is a thematic cartographic visualization, in which the areas of countries are modified in order to represent a given set of values, such as population, gross-domestic product, or other geo-referenced statistical data. Red-and-blue population cartograms of the United States were often used to illustrate the results in the 2000 and 2004 presidential elections. A geographically accurate map seemed to show an overwhelming victory for George W.~Bush; see Fig.~\ref{fig:uselection-geo}. The population cartograms effectively communicate the near even split, by deflating the rural and suburban central states. The rectilinear cartogram shows the correct distribution of red and blue squares, each representing one vote in the electoral college, but many characteristic shapes and adjacencies are compromised; see Fig.~\ref{fig:uselection-recti}. For example, Idaho and Washington are no longer neighbors, and the mirror-image shapes of New Hampshire and Vermont are lost. The balloon cartogram also shows the correct areas, but at the cost of distorted shapes and changes in above/below, left/right relationships; see Fig.~\ref{fig:uselection-newman}.

The challenge in creating a good cartogram is thus to shrink or grow the regions in a map so that they faithfully reflect the set of pre-specified area values, while still retaining their characteristic shapes, relative positions, and adjacencies as much as possible. In this paper we introduce
a new {\em circular-arc cartogram} model, where circular arcs can be used in place of straight-line segments, and corners of the polygons defining each country remain fixed.
Intuitively, a region that grows is inflated and becomes
 cloud-shaped, whereas a region that shrinks is deflated and becomes
 snowflake-shaped. 

Consider the circular-arc cartogram for the 2004 US presidential election: like a traditional cartogram, it also inflates densely populated states (which become cloud-shaped) and deflates sparsely populated ones (which become snowflake shaped); see Fig.~\ref{fig:uselection-circ}.  Note that the circular-arc cartogram preserves adjacencies, and the general shape of the states. Moreover, the circular-arc cartogram
makes it easy to see that nearly all blue states are densely populated and nearly all red states are sparsely populated, something that is not apparent in the rectilinear cartogram.  Finally, exceptions from this pattern are also easy to spot: Oregon is blue but sparse and North Carolina is red but dense. Of course, there is no such thing as a free lunch: the advantages of the circular-arc cartogram come at the expense of some cartographic errors, where accurate inflation and deflation cannot be guaranteed.

There are many design and implementation aspects that determine the
effectiveness of a cartogram. Here we consider four of the main aesthetic and computational criteria:
\begin{enumerate}
\item It is important that the cartogram is {\em readable}, in that it is possible to find every country in the map. Moreover, a readable cartogram makes is possible to visually answer approximate queries about the relative size of the shown countries.
\item It is
 important to ensure that the
cartogram keeps the underlying map structure
 {\em recognizable}. This criterion can be expressed by insisting that the country
 adjacencies in the original map and the cartogram remain
 unchanged. An even stronger version of this requirement is to ensure
 that the relative positions between pairs of countries (e.g., North-South, East-West) are not disturbed.
\item It is important that the cartogram faithfully represents the
given weight function. This criterion is often expressed by the {\em
  cartographic error}, defined as the absolute or relative difference
between the given weight and the area of a country. 
\item The {\em complexity} of a cartogram also impacts its effectiveness. Here, the complexity is often measured by the maximum number of vertices (or
edges) defining the boundary of any country in the cartogram. Highly
schematized cartograms use as few as three or four vertices per
country, while geographically more accurate and recognizable cartograms may
have arbitrarily high complexity.
\end{enumerate}

It is easy to see that there is no perfect method for generating cartograms, that is, there is no method that satisfies all of the main criteria. Most existing methods aim for no cartographic error and low complexity, while sacrificing recognizability (e.g., by allowing adjacencies to be modifies) and/or readability (e.g., by using arbitrary country shapes). 
 Circular-arc cartograms ensure readability by keeping the corners of the
countries undisturbed and easily convey the type of area changes by the
cloud-shape and snowflake-shape of the countries. They are also recognizable
as they retain all adjacencies and also preserve the relative positions of
countries. The complexity is exactly the same as that of the input map: a
highly schematized input map directly results in low complexity of the
resulting cartogram, which at the same time has the advantage that longer
edges allow for larger area changes and thus potentially lower cartographic
error. These advantages come at a cost: it is possible that a given map with
pre-specified areas cannot be realized as a circular-arc cartogram, and
determining whether such a realization exists is NP-hard. However, if we are
willing to tolerate moderate cartographic errors we can use a heuristic
algorithm, which, while not perfectly accurate, achieves many of the desired
areas in our real-world examples.

\subsection{Related Work}
   The problem of representing additional information on top of a
   geographic map dates back to the 19th century, and highly schematized rectangular cartograms can be found in the 1934 work of Raisz~\cite{raisz}.
 With rectangular cartograms it is not always possible to preserve all country adjacencies and realize all areas accurately~\cite{DBLP:conf/infovis/HeilmannKPS04,ks-rc-07}.
    Eppstein {\em et al.} studied area-universal rectangular layouts
    and characterized the class of rectangular layouts for which all
    area-assignments can be achieved with combinatorially equivalent
    layouts~\cite{DBLP:conf/compgeom/EppsteinMSV09}. If the
    requirement that rectangles are used is relaxed to allow the use of
rectilinear regions then de Berg {\em et al.}~\cite{DBLP:journals/dm/BergMS09} showed that all adjacencies can be preserved and all areas can be realized with 40-sided regions. 
In a series of papers the polygon complexity that is sufficient to realize any rectilinear cartogram was decreased from 40 over 34 corners~\cite{kn-odpgwsfa-07}, 12 corners~\cite{bv-ocwcf-11}, 10 corners~\cite{abfgkk-lapcgr-11} down to 8 corners~\cite{abfkku-ccwoc-11}, which is best possible due to the earlier lower bound of 8-sided regions~\cite{ys-fgd2rm-93}.

 More general cartograms, without restrictions to rectangular or
    rectilinear shapes, have also been studied.  Dougenik {\em et al.} introduced a method based on force fields where the map is divided into cells and every cell has a force related to its data value which affects the other cells~\cite{PROG:PROG75}.
    Dorling used a cellular automaton approach, where regions exchange cells until an equilibrium has been achieved, i.e., each region has attained the desired amount of cells~\cite{Dorling}. This technique can result in significant distortions, thereby reducing readability and recognizability.
    Keim {\em et al.} defined a distance between the original map and the cartogram with a metric based on Fourier transforms, and then used a scan-line algorithm to reposition the edges so as to optimize the metric~\cite{DBLP:journals/tvcg/KeimNP04}.
    Edelsbrunner and Waupotitsch generated cartograms using a sequence of homeomorphic deformations and measured the quality  with local distance distortion metrics~\cite{DBLP:journals/comgeo/WelzlEW97}.
    Kocmoud and House~\cite{House:1998:CCC:288216.288250} described a technique that combines the cell-based approach of Dorling~\cite{Dorling} with the homeomorphic deformations of Edelsbrunner and Waupotitsch~\cite{DBLP:journals/comgeo/WelzlEW97}.

A popular method by Gastner and Newman~\cite{Gastner} projects the original map onto a distorted grid, calculated so that cell areas match the pre-defined values. This method relies on a physical model in which the desired areas are achieved via an iterative diffusion process. Flow moves from one country to another until a balanced distribution is reached, i.e., the density is the same everywhere. The cartograms produced this way are mostly readable and have no cartographic error. However, some countries may be deformed into shapes very different from those in the original map, and the complexity of the polygons can increase significantly. 

This brief review of related work is woefully incomplete; a %
survey by Tobler~\cite{Tobler04thirtyfive} provides a more comprehensive overview.

\subsection{Our Contributions}

Our model combines aspects of existing cartogram types, but at the
same time tries to avoid some of the common shortcomings. By pinning the vertices at their input positions and only modifying edge shapes, regions are not displaced and we avoid strong positional distortions that are common, e.g., in the popular diffusion cartograms. On the other hand, the shapes of the regions are not as severely schematized as in rectangular or rectilinear cartograms and recognizability of characteristic shapes is preserved, at least for moderate area changes. The use of the inflation/deflation metaphor makes is possible to immediately recognize regions with positive/negative area changes.

Our results in this paper are as follows. In Section~\ref{sec:model} we formally introduce the circular-arc cartogram model and state the associated algorithmic problem. In Section~\ref{sec:complexity} we show that the circular-arc cartogram problem is NP-hard. In Section~\ref{sec:algorithm} we describe a first heuristic algorithm using network flow and the straight skeleton to minimize the cartographic error in  circular-arc cartograms. %
In Section~\ref{sec:conclusion} we summarize our results and describe several  open problems.

\section{Model}\label{sec:model}
Geometrically, a map of countries or administrative regions is a subdivision $S$ of the plane into a set of disjoint regions or \emph{faces} $\mathcal{F} = \{f_1, \dots, f_n\}$. 
In our model we assume that each face is a simple polygon. 
The topological structure of the map can be described by its \emph{face graph} or \emph{dual graph} $G$, which contains a vertex for each face and an edge between adjacent faces.
In order to construct a cartogram of $S$, we additionally need to specify a weight vector $t=(t_1, \dots, t_n)$, where for each $i=1, \dots, n$ the value $t_i$ is the target area of face $f_i$ in the cartogram.
An \emph{accurate} cartogram of the input pair $(S,t)$ is a subdivision $S'$ that is homeomorphic to $S$ and in which the area of every face $f_i$ equals its given weight $t_i$.

In this paper, we are interested in the special class of \emph{circular-arc cartograms}, i.e., cartograms that can be obtained from the input $S$ by bending each polygon edge $e$ into a circular arc whose endpoints coincide with the endpoints of $e$. 
No two circular arcs are allowed to cross, but we may allow that two arcs touch. 
Bending an edge between two faces $f_i$ and $f_j$ has the effect of transferring a certain area from one face to the other. 
This exchange of area between faces can be seen as a discrete diffusion process similar to the model of Gastner and Newman~\cite{Gastner}. 
The algorithmic problem in creating a circular-arc cartogram is thus to compute a bending radius for each edge of the input subdivision so that the resulting circular-arc subdivision $S'$ remains topologically equivalent to the polygonal input subdivision $S$ and each face $f_i$ has area $t_i$.
We define a \emph{bending configuration} of~$S$ to be an assignment of a bend radius (including radius $r=\infty$ to represent straight-line arcs) to each edge of $S$. A bending configuration is \emph{valid} if no two circular arcs cross and the input topology of $S$ is preserved. We say that a cartogram $S'$ is a \emph{strong} circular-arc cartogram if for every region $f_i$ with a net decrease (increase) in area no incident edge is bent outward (inward). Otherwise we call the cartogram a \emph{weak} circular-arc cartogram. An immediate consequence of strong cartograms is that edges bounding two regions with the same sign of area change must remain straight.

In real-world maps there are often regions (e.g., oceans or seas) whose target area in the cartogram is unspecified. 
Our model allows sea faces in $S$ with no specified target area. 
Note that if there is a single sea face then its target area change is implicitly given by the sum of the target area changes of the other faces.

The \textsc{Circular-Arc Cartogram} (CAC) decision problem then is:
  \begin{problem} [\textsc{CAC}]%
  \label{prob:cac}
  Given a planar polygonal subdivision $S$ and a weight vector $t$, is there a valid bending configuration so that the resulting subdivision $S'$ is an accurate circular-arc cartogram, i.e., all face areas in $S'$ comply with $t$?
  \end{problem}

While the decision version is mainly of theoretical interest, there is also a corresponding optimization version of \CAC.
Here the algorithmic problem is to compute a bending configuration that minimizes the \emph{cartographic error}, i.e., the sum of the differences between the target areas and the actual areas of all faces.
In Section~\ref{sec:complexity} we show that CAC is NP-hard and in Section~\ref{sec:algorithm} we describe a heuristic algorithm that successfully minimizes the cartographic error in practice.

\section{NP-hardness}\label{sec:complexity}
	First note that positive as well as negative CAC instances can be constructed easily: 
	Only polygons whose vertices are cocircular can be made arbitrarily small by bending edges; all other polygons have some positive lower bound on their area in a circular-arc cartogram.
	Hence, for example, no simple non-convex polygon can attain area close to $0$ by replacing straight edges with circular arcs. 
	On the other hand, any subdivision with a target area vector that contains the exact initial face areas is a positive instance.

    \begin{theorem}
    \CAC is NP-hard.
    \end{theorem}

\begin{proof}
	Our reduction is from the NP-complete problem \textsc{Planar Monotone 3-Sat}~\cite{bk-obspp-10}. 
	This problem is a special variant of the \textsc{Planar 3-Sat} problem~\cite{l-pftu-82}: 
	We are given a Boolean formula $\varphi$, in which every clause consists of three literals. 
	Each clause, however, must be monotone, i.e., it may contain either only positive or only negative literals. 
	The planarity of the formula refers to the planarity of the associated bipartite variable-clause graph $G_\varphi$ (with a vertex for every clause and variable of $\varphi$ and an edge between a variable vertex and a clause vertex if and only if the variable appears in the clause).
	It is known that for every instance of \textsc{Planar Monotone 3-Sat} the graph $G_\varphi$ can be drawn in a planar rectilinear fashion by placing the variable vertices on a horizontal line, the positive clauses above that line, and the negative clauses below;	see Fig.~\ref{fig:planmon3sat}.%
	
	\begin{figure}[htb]
		\centering
			\includegraphics[scale=1]{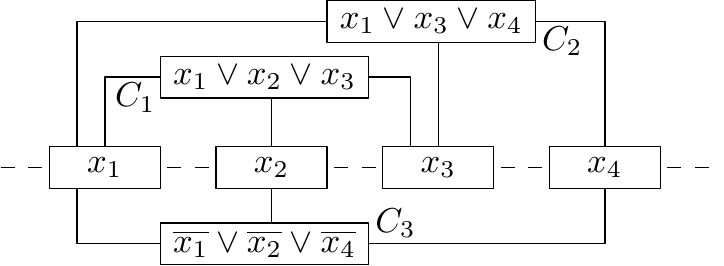}
		\caption{A planar rectilinear drawing of a \textsc{Planar Monotone 3-Sat} instance.}
		\label{fig:planmon3sat}
	\end{figure}
	
	Our reduction constructs a subdivision $S_\varphi$ for the Boolean formula $\varphi$ that resembles the general structure of the rectilinear drawing of $G_\varphi$. The weight vector $t_\varphi$ is chosen so that $S_\varphi$ can be transformed into a valid circular-arc cartogram if and only if $\varphi$ is satisfiable. The subdivision consists of three types of gadgets: the \emph{variable}, \emph{literal}, and \emph{clause} gadgets, which we describe below.
	
	A basic building block in all three gadgets is a triangle with target area~$0$. It is easy to verify that there are exactly three configurations that realize a $0$-area circular-arc triangle, all of which consist of circular arcs of the unique circle defined by the three points; see Fig.~\ref{fig:area0}. This building block is used to control the possible shapes of regions in the cartogram.
	
	\begin{figure}[htb]
		\centering
			\includegraphics[width=\columnwidth]{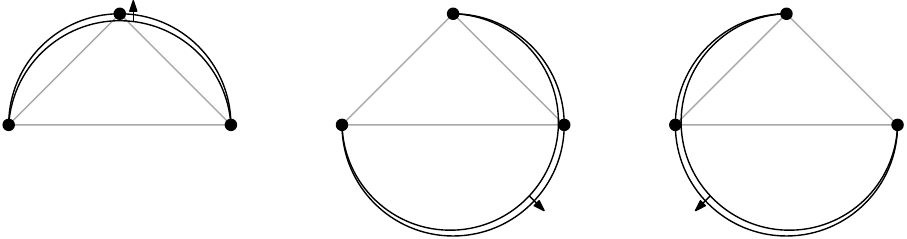}
		\caption{The three  possibilities to realize a circular-arc triangle with area~$0$.}
		\label{fig:area0}
	\end{figure}
	
	\paragraph{Variable gadget} The variable gadget consists of a horizontal row of rectangles with height $4$ and width $2$, except for some taller rectangles in between of height $5$ that serve as \emph{connectors} to the literal gadgets; see Fig.~\ref{fig:var-gadget}. 
	With the exception of the connector rectangles, all rectangles are enclosed on their two short sides by skinny triangles with a base side of length $2$. 
	These triangles have target area $0$.
	They are designed so that two of the three possibilities to achieve area $0$ would require edges to become circular arcs that pass beyond some of the input vertices of the rectangles.
	Hence only a single configuration remains feasible. 
	This immediately fixes the shape of the rectangles' short edges by bending them slightly outward and increases the area of each rectangle by the area of two circular segments.
	We define the area of the circular segment thus attached to each rectangle as $c_1$. 
	We also need a scaled-down version of this triangle with base length $1$ instead of $2$ whose corresponding circular segment thus has area $c_2 = c_1/4$.
	
	\begin{figure}[htb]
			\centering
			\subfloat[Input subdivision]{\includegraphics[scale=.5,page=2]{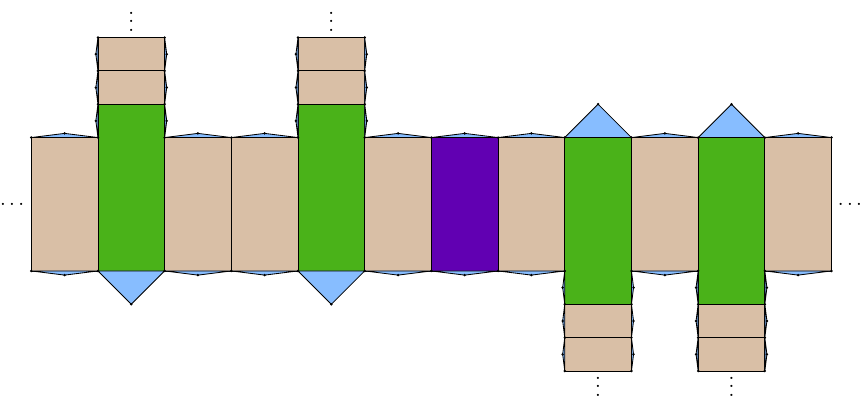}} \hfill
			\subfloat[Variable set to \emph{true}]{\includegraphics[scale=.5,page=3]{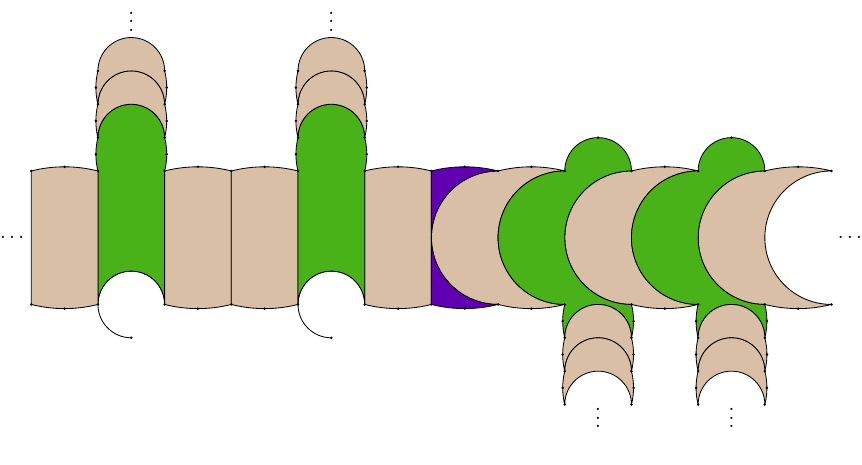}} \hfill
			\subfloat[Variable set to \emph{false}]{\includegraphics[scale=.5,page=4]{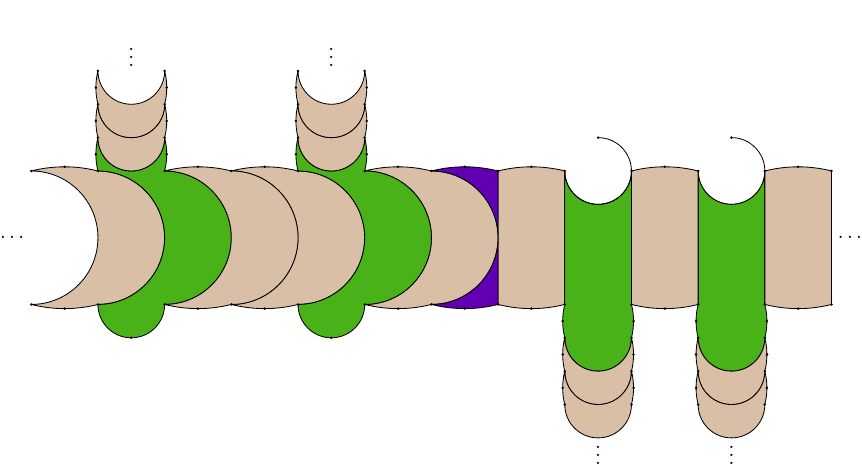}}
		\caption{Variable gadget. The central decision rectangle is shown in purple, connector rectangles in green, $0$-area triangles in blue and the remaining rectangles in brown.}
		\label{fig:var-gadget}
	\end{figure}
		
	There is one special \emph{decision rectangle} (purple) in the center of the gadget. 
	The target area change of this rectangle is set to $2c_1-2\pi$, where $2\pi$ is exactly the area of a half-circle with radius $2$.
	All other rectangles of height $4$ have a target area change of $2c_1$, i.e., they can be extended by the two circular segments at their short sides but otherwise want to keep their area constant. 
	Finally, the taller connector rectangles (which actually consist of six vertices) are adjacent to a literal gadget on one of their short sides  (indicated by dots in Fig.~\ref{fig:var-gadget}) and to a right triangle on the other short side. 
	This triangle has target area $0$, but other than the skinny triangles described before, all three possible $0$-area configurations in Fig.~\ref{fig:area0} are feasible. 
	The area change of the connector rectangles is $2c_2$, the area gained from the two small skinny triangles adjacent to the left and right sides of the length-1 edges that stick out of the variable row.
	
  Let us consider the purple decision rectangle in the center, with its two short edges fixed by the shape of the attached skinny triangles.
	If one of its long edges is bent inside the rectangle as exactly a half-circle and the opposite edge remains a straight-line segment, then the specified area constraint is satisfied. 
	It is, however, geometrically impossible to achieve the given target area by bending both edges simultaneously inside the rectangle like in a concave lens.
	Hence we can use the two possible configurations of the decision rectangle to encode the two truth values of the variable; see Fig.~\ref{fig:var-gadget}b and~\ref{fig:var-gadget}c.
	Since by pulling one long edge inside the decision rectangle the area of the adjacent rectangle in the gadget enlarges, that adjacent rectangle must in turn pull the opposite long edge inwards by the same amount. 
	So the semi-circle arcs propagate, similar to negative air pressure in a physical model, on one side of the gadget, namely that side whose connecting literals evaluate to \emph{false} in the current state.
	
	It remains to describe the behavior of the connector rectangles. 
	Since the long edges are bent into half-circles and no two edges of the subdivision may cross, the right triangle attached to the connector rectangle must be in the state that forms a half-circle and thus increases the area of the connector rectangle. 
	In order to balance this area increase, the opposite short edge must be bent inwards and form an identical half-circle.
	This gives us a means to transmit the negative pressure from the center of the variable towards all literals that evaluate to \emph{false}.
	
	There is no negative pressure on the positive side of the variable gadget, i.e., the side whose literals evaluate to \emph{true}. 
	Hence the long edges of the rectangles on this side can remain straight, and there are two possible configuration for the short edges of each connector rectangle, one of which pushes the half-circles towards the literal gadgets rather than pulling them away as it is the case on the negative side. 
	
	\paragraph{Literal gadget} The main task of the literal gadget is to maintain and transmit the truth state that is found at the variable gadget towards the clause gadget.
	The gadget can be seen as a pipe composed of chains of rectangles that connects variable and clause.
	In the pipe the truth value is transmitted by pulling or pushing the long edges of the rectangles into half-circles similarly as in the variable gadget.
	Three literal gadgets are depicted (together with a clause gadget) in Fig.~\ref{fig:clause-gadget}.
	
	\begin{figure}[htb]
	\centering
			\subfloat[Input subdivision]{\includegraphics[scale=.66,page=5]{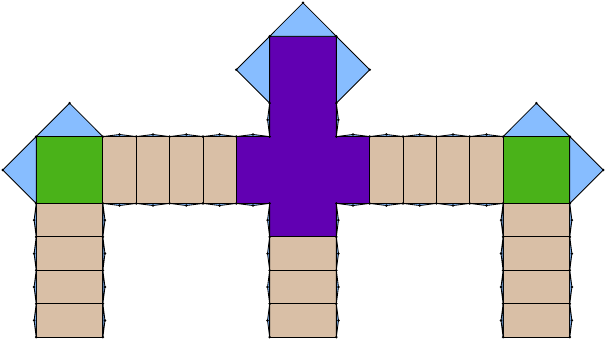}}\\
			\subfloat[Three false literals]{\includegraphics[scale=.66,page=6]{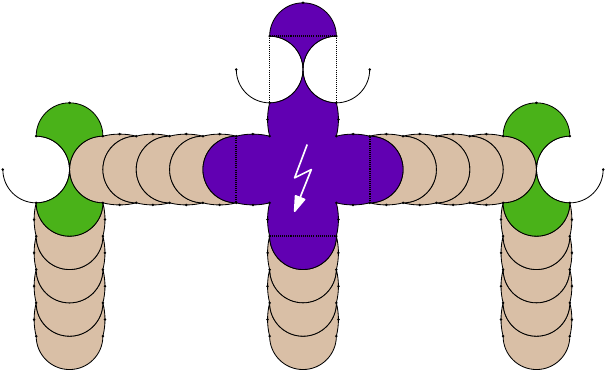}}
			\hfill
			\subfloat[One true and two false literals]{\includegraphics[scale=.66,page=7]{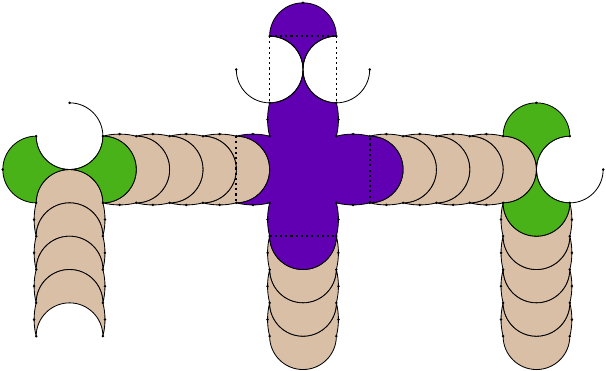}}
		\caption{Cross-shaped clause gadget with three literal gadgets.}
		\label{fig:clause-gadget}
	\end{figure}
	
	There is one notable difference from the transmission of the truth value in the variable gadget since two of the incoming literals for each clause make a turn of 90$^\circ$. 
	The turn is realized by a square of side length $2$ with target area $4$ and two right triangles with target area $0$ (as those  in the connector rectangles of the variable gadget). 
	Since the two right triangles are placed on adjacent sides of the square, one of them must bend to the outside of the square while the other one must bend to the inside. 
  If the literal is in state \emph{false} (left and right in Fig.~\ref{fig:clause-gadget}b) and the half-circles are pulled towards the variable gadget, then both the left and right edges of the square are bent inward while the top and bottom edges are bent outward. 
	This is exactly what is needed to transmit negative pressure to the horizontal part of the literal gadget.
	For a literal in state \emph{true} we observe exactly the opposite behavior. Fig.~\ref{fig:clause-gadget}c shows a \emph{true} literal on the left and a \emph{false} literal on the right.
	
	\paragraph{Clause gadget} The clause gadget consists of a cross-shaped rectilinear \emph{clause polygon} joining the three incoming literal gadgets; see Fig.~\ref{fig:clause-gadget}.
	In its top part there are three right triangles with target area~$0$. 
	The target area increase of the clause polygon is $8c_2$, the area increase caused by the eight skinny triangles attached to some of its edges.
	Note that of the three right triangles at most two can simultaneously bend as half-circles inside the polygon, while they all can bend to the outside independently.
	As long as one of the incoming literals is \emph{true}, i.e., it pushes a half-circle inside the clause polygon, the three triangles in the top part can balance the area change of the clause polygon caused by any other combination of the remaining two literals; see Fig.~\ref{fig:clause-gadget}c. %
	However, if all three literals are \emph{false}, the area of three half-circles is added to the clause region (indicated by dotted line segments). 
	Consequently, the area of three half-circles must be removed from the clause region, but at most two half-circles can be removed by the right triangles; see Fig.~\ref{fig:clause-gadget}(b). 
	This shows that the area requirement of the clause polygon can be realized if and only if the clause evaluates to \emph{true} in the given truth assignment.

	\paragraph{Reduction} 
	From the construction of the gadgets it follows that if the Boolean formula $\varphi$ has a satisfying variable assignment, then the subdivision $S_\varphi$ and the weights $t_\varphi$ are a positive instance of \CAC. On the other hand, we can immediately obtain a satisfying truth assignment for the variables of $\varphi$ from a valid circular-arc cartogram of $S_\varphi$. 
	The vertices of the subdivision $S_\varphi$ all lie on a grid of polynomial size and the target weights are either $0$ or can be encoded algebraically in polynomial space. 
	This concludes the proof.%
\end{proof}

We note that all vertices in the subdivision $S_\varphi$ either belong to triangles or have degree at least~3. Thus the complexity of $S_\varphi$ cannot be decreased further and the hardness result continues to hold for maps that are minimal in that sense.

\section{Heuristic Method for Computing Circular-Arc Cartograms}\label{sec:algorithm}

Here we describe a versatile heuristic method for generating circular-arc cartograms based on network flows and polygonal straight skeletons. In practice we may assume that our input map is already a simplified or even schematized map that retains the characteristic shapes of the countries but at the same time strongly reduces the polygon complexities. The more simplified the shapes the longer the edges and, consequently, the larger the potential area changes that can be realized by bending the edges. 
Buchin \emph{et al.}~\cite{bms-msswaug-11} or de Berg \emph{et al.}~\cite{bks-ass-95} described suitable algorithms for computing topologically correct subdivision simplifications and schematizations.

Recall that in our model a map is a subdivision $S$ of the plane into a set of disjoint faces, $\mathcal{F} = \{f_1, \dots, f_n\}$, where each face is a simple polygon. The topological structure of the map is described by its dual \emph{face graph} $G$, which contains a vertex $v_i$ for each face $f_i$ and an edge $\{v_i,v_j\}$ between adjacent faces $f_i$ and $f_j$. 
Here we convert $G$ into a directed graph: for any two adjacent countries in $S$ the corresponding vertices in $G$ are connected with two edges, one for each direction.

The initial face areas are described by the vector $a=(a_1, \dots, a_n)$ and the target areas are given by the vector $t=(t_1, \dots, t_n)$. 
Without loss of generality, we can assume that both vectors are normalized, i.e., $\sum_{i=1}^n a_i = \sum_{i=1}^n t_i = 1$. This means that the total area of the map remains the same. From $a$ and $t$ we can obtain the vector $\Delta = (\Delta_1, \dots, \Delta_n)$ of \emph{desired area changes}, where $\Delta_i = t_i - a_i$ for each $i=1, \dots, n$. Note that $\sum_{i=1}^n \Delta_i = 0$.

The goal of our algorithm is to compute a valid bending configuration in which the resulting face areas $b=(b_1, \dots, b_n)$ are as close to the given target areas $t$ as possible. More precisely, we aim to minimize the error $\sum_{i=1}^n |b_i - t_i|$.

\subsection{A Network Flow Model for Circular-Arc Cartograms}

We use the directed face graph $G$ to define a flow network in which the flow along an edge $e=(u,v)$ corresponds to the area exchange from the face vertex $u$ to the face vertex $v$. 
We define the capacity $c(e)$ to be equal to an area that can be safely transferred from $u$ to $v$. To compute valid capacities we use the geometry of the face polygons that specify the countries. If we want to restrict ourselves to strong circular-arc cartograms we set $c(e)=0$ for all edges~$e$ between two regions $f_i$ and $f_j$ for which $\Delta_i \cdot \Delta_j \ge 0$; for weak cartograms there is no such restriction.

The {\em straight skeleton} of a simple $m$-edge polygon, $P$, is made of straight-line segments and partitions the interior of $P$ into $m$ disjoint regions, each corresponding to exactly one edge of $P$~\cite{DBLP:journals/jucs/AichholzerAAG95}. The straight skeleton is similar to the medial axis but does not require parabolic curves and can be efficiently computed in subquadratic time~\cite{DBLP:journals/algorithmica/ChengV07}. %
Because the straight skeleton partitions a polygon into disjoint regions, we can define a ``safe'' bending limit for each edge of the polygon by requiring that the circular arcs remain inside their skeleton regions; see Fig.~\ref{fig:skeleton}. This guarantees that no two circular arcs cross. For each edge $e=(u,v)$ we can thus define the capacity $c(e)$ as the maximally transferable area from face $u$ to face $v$ subject to the constraint that every circular arc on the boundary between $u$ and $v$ remains inside its skeleton regions. The capacities are by definition static and independent of each other. We note that there is still room for enlarging the capacities over this definition, e.g., by considering to remove some degree-2 vertices and consequently merge their incident boundary edges and their skeleton regions. This yields longer boundary edges that allow larger arcs with larger transferable areas. 

\begin{figure}[htb]
    \includegraphics[width=\columnwidth]{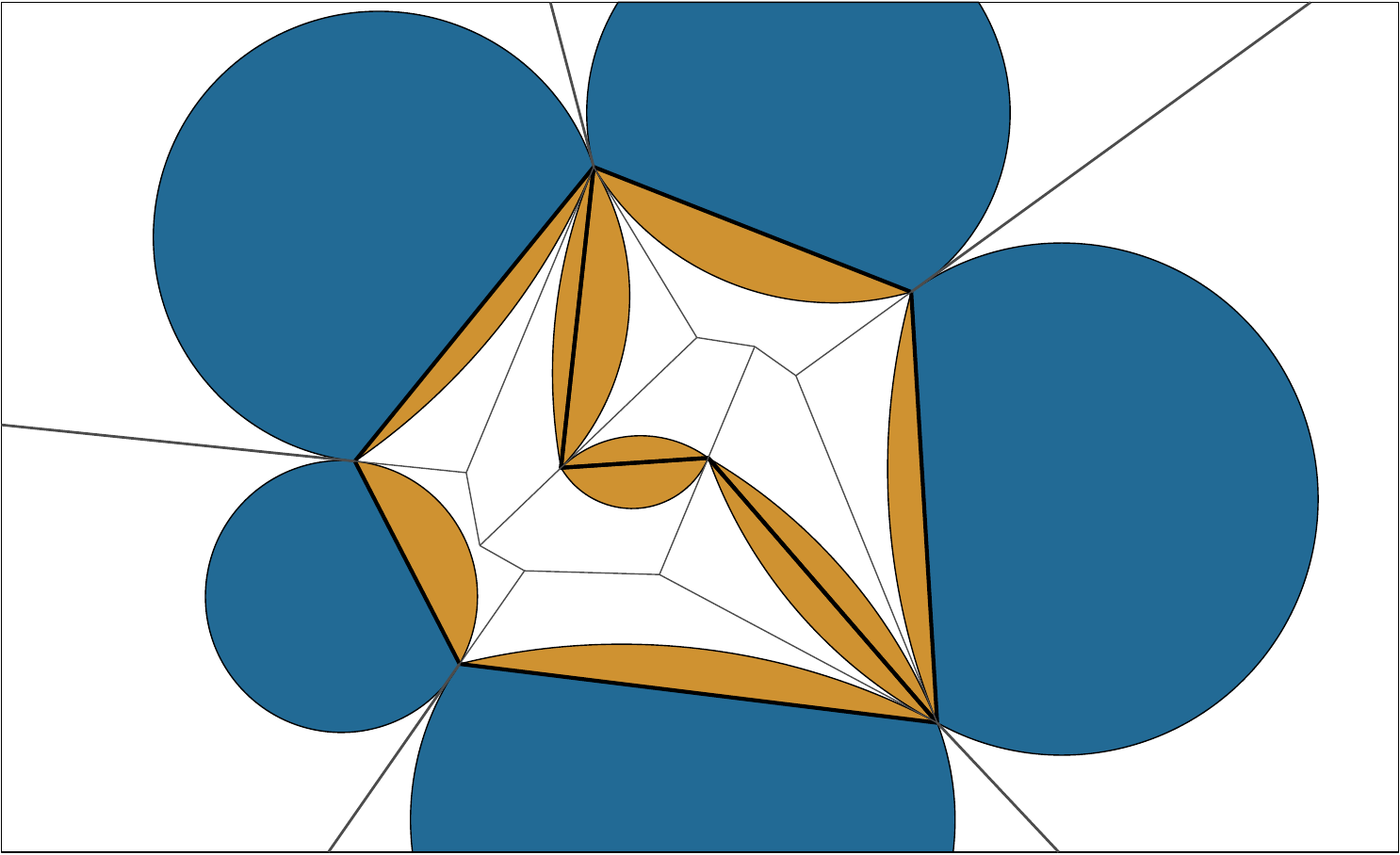}
	\caption{The straight skeleton of two adjacent polygons and the maximally realizable circular-arcs within the safe bending limits of each edge.}
	\label{fig:skeleton}
\end{figure}

Once we have computed a set of valid edge capacities for~$G$, we create a new vertex $v_i'$ for every vertex $v_i$ in $G$. If $\Delta_i > 0$ we make $v_i'$ a source vertex and add the edge $(v_i',v_i)$ with capacity $c(v_i',v_i) = \Delta_i$ to $G$; otherwise if $\Delta_i < 0$ we make $v_i'$ a sink vertex and add the edge $(v_i,v_i')$ with capacity $c(v_i,v_i') = - \Delta_i$ to $G$. Let $S$ be the set of sources and $T$ the set of sinks.

    The quadruple $\mathcal{N}=(G,c,S,T)$ now forms a multiple-source multiple-sink flow network, which is planar since the original face graph of $S$ was planar.
    If a maximum flow in $\cal N$ with a value of $D=\sum_{\Delta_i > 0} \Delta_i$ can be found, we know that all target areas can be achieved.
    Furthermore, even if the maximum flow has a value of less than $D$, it still corresponds to a bending configuration that minimizes the cartographic error $\sum_{i=1}^n |b_i - t_i|$ under the given safety constraints for the circular arcs.

\begin{figure}[h]
	\centering
		\includegraphics[width=\columnwidth]{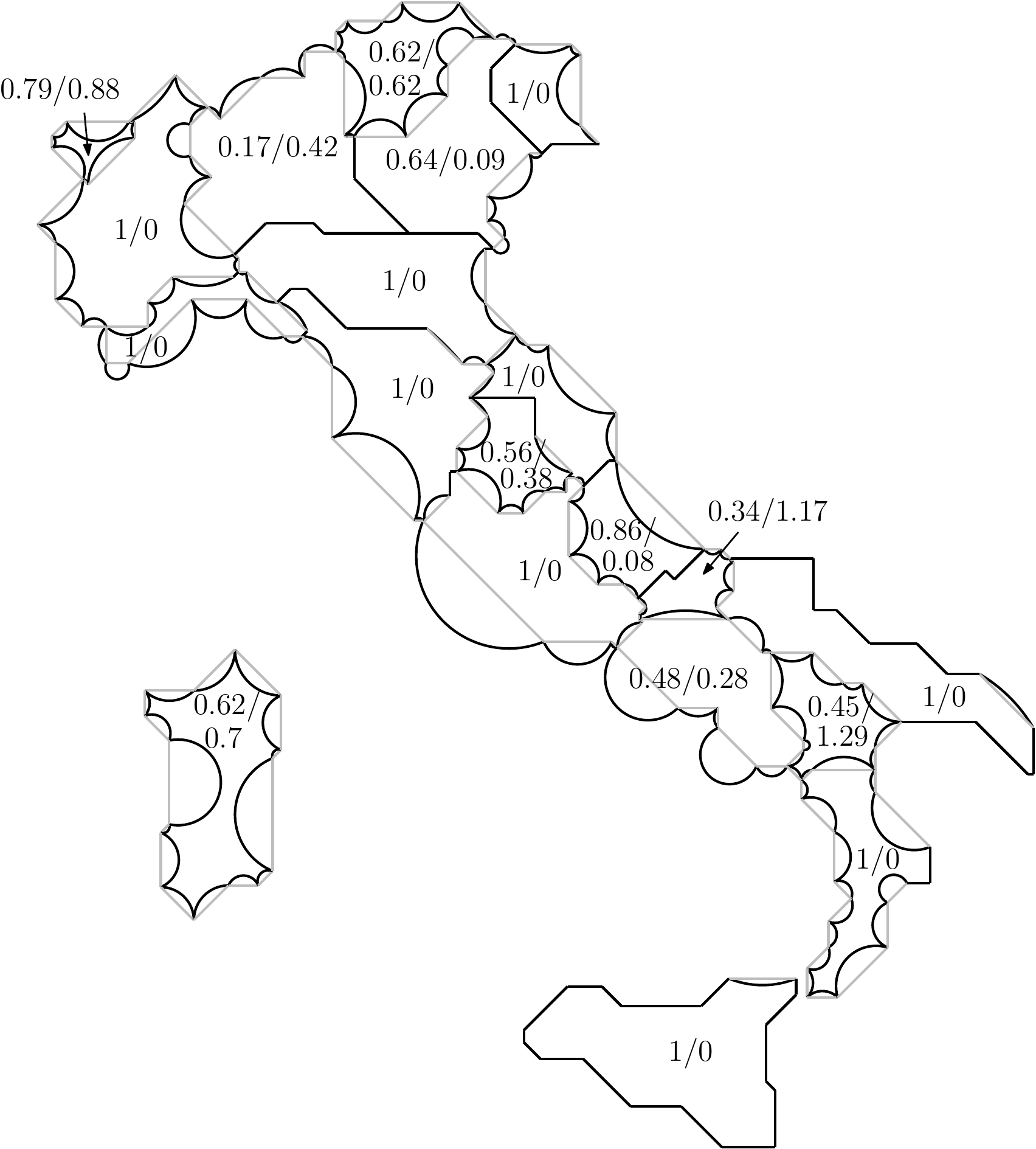}
	\caption{Cartogram of the population in Italy; the first number indicates the success rate, the second number the relative cartographic error.}
	\label{fig:ital-pop}
\end{figure}

\begin{figure}[h]
	\centering
		\includegraphics[width=\columnwidth]{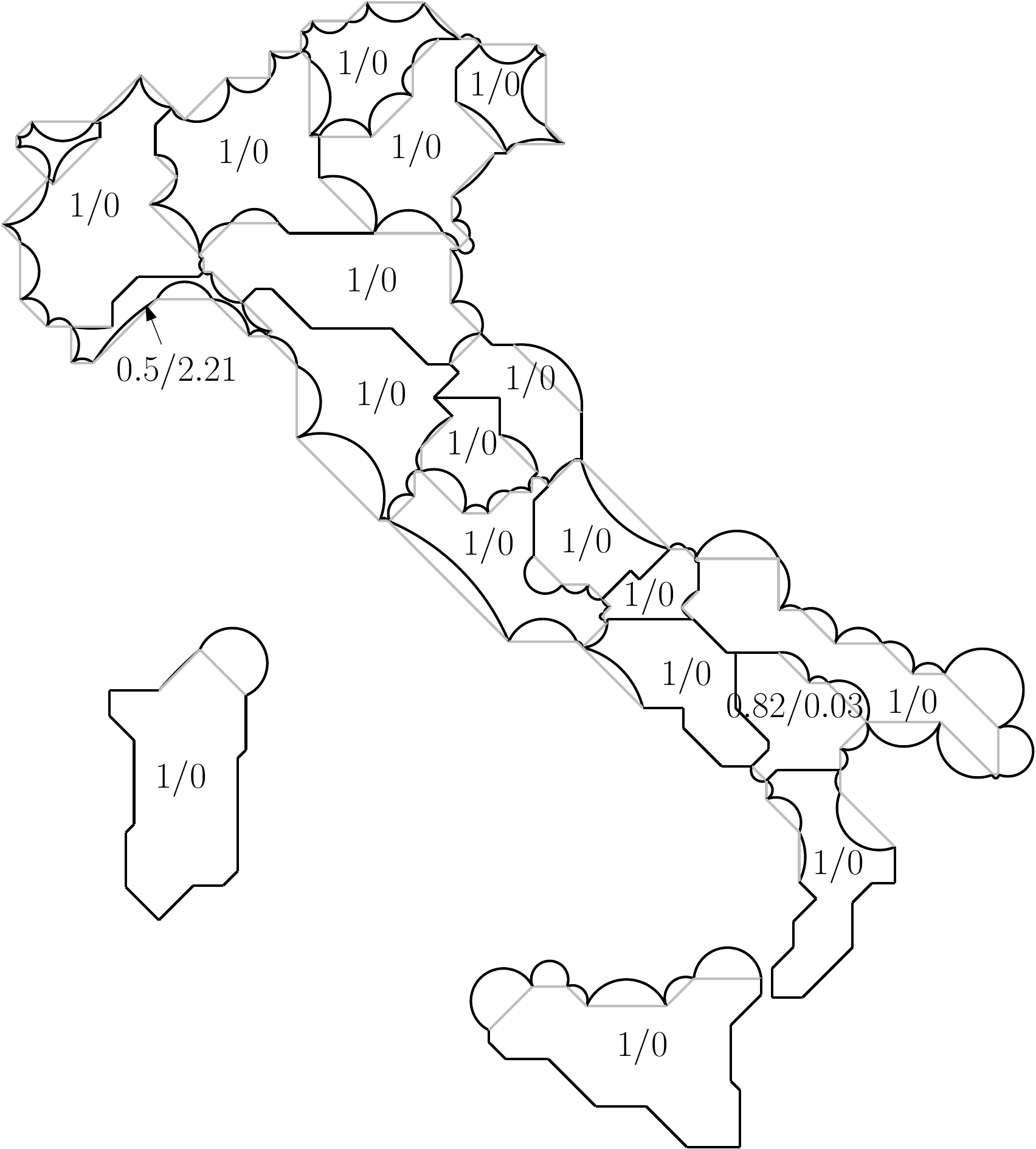}
	\caption{Cartogram of the agricultural use area in Italy; the first number indicates the success rate, the second number the relative cartographic error.}
	\label{fig:ital-agri}
\end{figure}
The expected running time for computing the straight skeleton of a $k$-vertex polygon is $O(k \log^2 k)$~\cite{DBLP:journals/algorithmica/ChengV07}. So if the input subdivision $S$ consists of $n$ faces with $N$ vertices in total, we can compute the straight skeletons in $O(N \log^2 N)$ expected time in total. 
To solve the multiple-source multiple-sink maximum-flow problem in our flow network based on the planar face graph of $S$ we can use the recent $O(n \log^3 n)$-time algorithm of Borradaile \emph{et al.}~\cite{bkmnw-mmmfdpgnt-11}.

\subsection{Implementation and Results}

We implemented a prototype of our method in C++ using the CGAL library~\cite{cgal} for computing the straight skeletons and Boost~\cite{boost} for solving the max-flow problem. 
In this section we present four examples of circular-arc cartograms produced with our implementation, which currently supports only weak circular-arc cartograms.
As input subdivisions we used octilinear and rectilinear 
schematized maps generated with the algorithm of Buchin et al.~\cite{bms-msswaug-11} for area-preserving subdivision schematization.
We note here that keeping the vertex positions fixed in our method only makes sense if these vertices are actually characteristic corners of the original shape. This is an interesting problem in its own right, which could be addressed with a shape simplification method that identifies and retains characteristic points, but out of scope in this paper. 
To demonstrate our circular-arc approach, we can assume that the vertices have been chosen in a meaningful way so the polygonal shapes represent the corresponding countries well. 
In the appendix, we present four additional examples using a manually simplified input map.

For each example below, we measure both the \emph{success rate}, which is defined as $(a_i-b_i)/\Delta_i$, i.e., the relative achieved area change, and the relative \emph{cartographic error}, which is defined as $|b_i-t_i|/t_i$~\cite{ks-rc-07}. We show the input polygons in gray and overlay the circular-arc cartogram and label each country with the pair (success rate, cartographic error).

Figures~\ref{fig:ital-pop} and~\ref{fig:ital-agri} show cartograms of the regions in Italy. 
Figure~\ref{fig:ital-pop} represents the population distribution in Italy\footnote{2010 population data from
\url{http://demo.istat.it/pop2010}.} and Figure~\ref{fig:ital-agri} represents the agricultural use areas in each region\footnote{2010 superficie agricola utilizzata data from Noi Italia 2012, \url{http://www3.istat.it/dati/catalogo/20120215_00/Noi_Italia_2012.pdf}}. 
This is an example of a map where our algorithm performs well. %
The average success rate in Figure~\ref{fig:ital-pop} is 0.78 and the average cartographic error is 0.3. In Figure~\ref{fig:ital-agri} the average success rate is as high as 0.97 and the average error is 0.11; here only two regions have non-zero error.
In the case of Italy most regions have access to the external sea face where the maximum size of circular arcs is less restricted. Moreover, the desired area changes are relatively moderate. With the removal of a few degree-2 vertices, i.e., a further simplification of the input subdivision, we could improve the area accuracy in Figure~\ref{fig:ital-pop} even more (in Sardegna or Campania) without the need of displacing vertices.

\begin{figure}[h]
	\centering
	\includegraphics[width=\columnwidth]{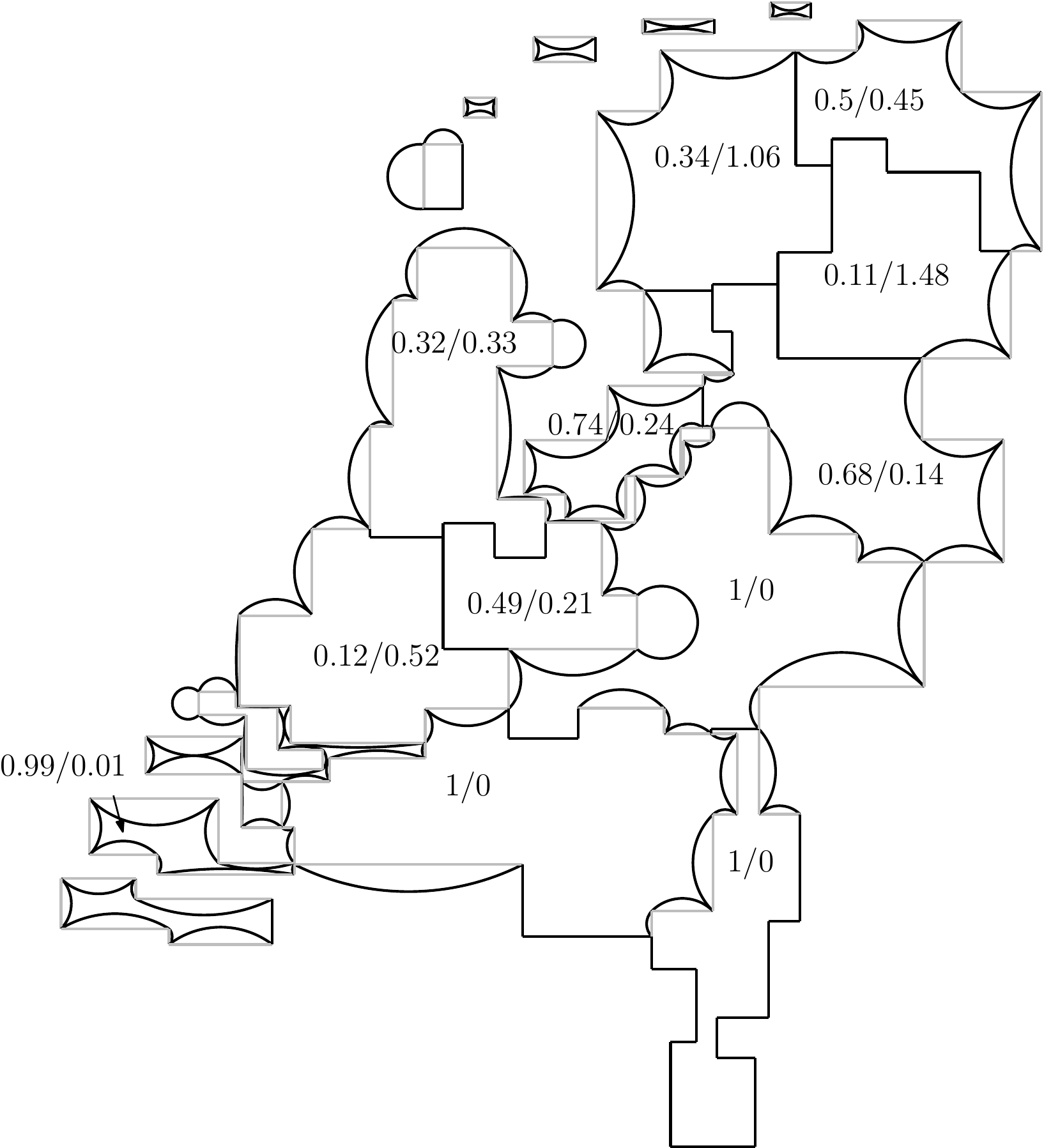}
	\caption{Cartogram of the population in the Netherlands.
}
	\label{fig:ned_pop}
\end{figure}

\begin{figure}[h]
	\centering
	\includegraphics[width=\columnwidth]{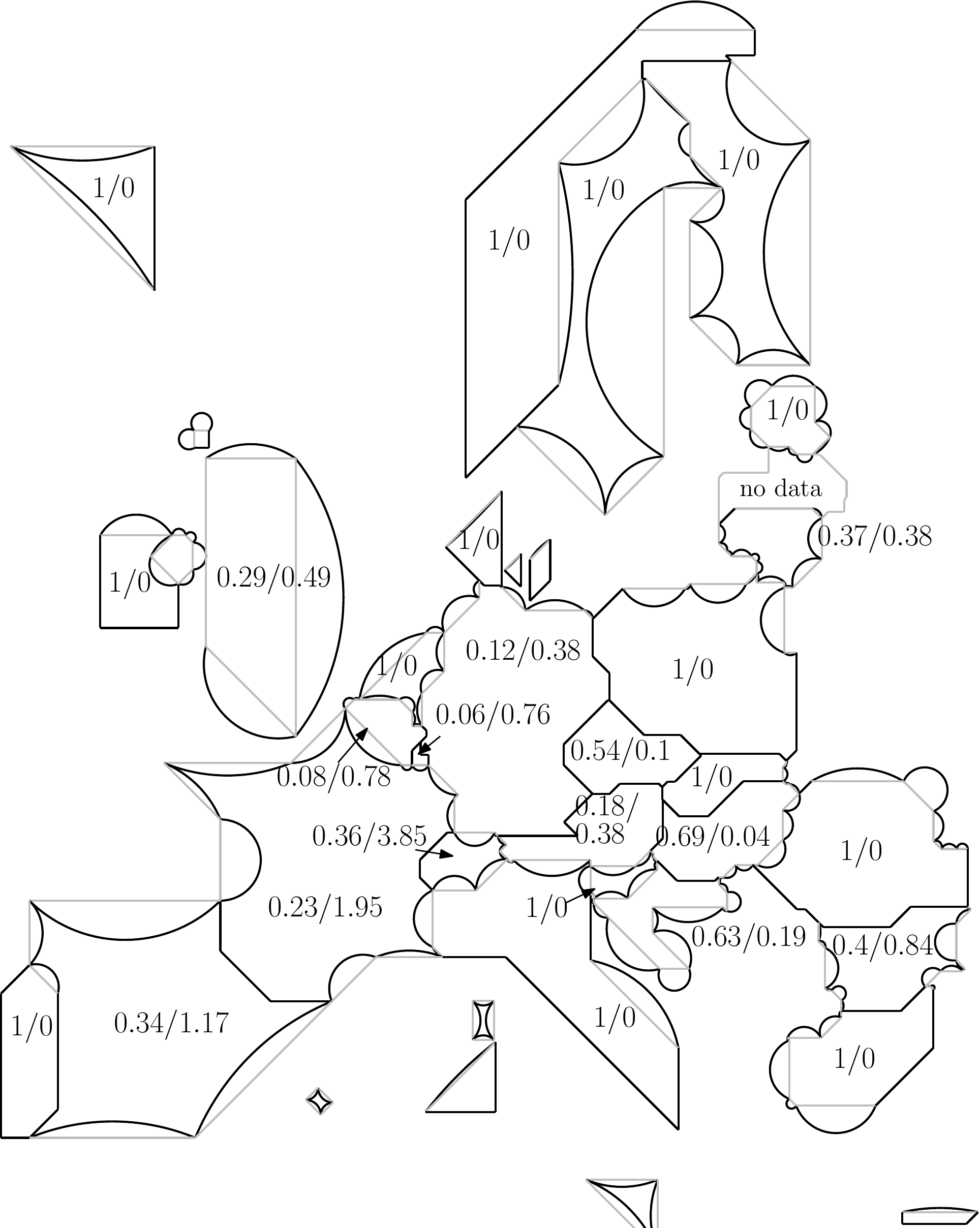}
	\caption{Cartogram of the length of the main roads in Europe.
}
	\label{fig:europe_mainroads_img}
\end{figure}

Figure~\ref{fig:ned_pop} shows a cartogram for the population distribution in the 
  Netherlands\footnote{2004 population data from \url{http://en.wikipedia.org/wiki/Ranked_list_of_Dutch_provinces}.}.
This cartogram is based on a rectilinear rather than an octilinear schematization.
  The Netherlands are quite unevenly populated: for example the three provinces of Noord-Brabant, Zuid-Holland and Noord-Holland (containing all important urban areas), contribute more than half of the Dutch population. This imbalance between south and
  north can bee seen well in the cartogram. The regions of the metropolitan south are cloud-shaped, while the northern rural areas are more snowflake-shaped.
The imbalance in the data leads to slightly worse performance than in the previous example of Italy; the average success rate is 0.61 and the average cartographic error is 0.37. Further simplification of polygons and potentially some vertex displacements will help to increase the area accuracies.

Figure~\ref{fig:europe_mainroads_img} shows a cartogram based on the length of main roads per country\footnote{2012 road length data from \url{http://ec.europa.eu/transport}}. %
All regions in both outputs are recognizable. We can easily identify the different countries as the overall shapes are not very distorted. Even the aspect ratios of most regions remain mostly unchanged. The length of all borders is at least as big as in the input, so we obtained an improved readability of adjacencies.

The results in Figure~\ref{fig:europe_mainroads_img} show an average success rate of 0.69 and an average cartographic error of 0.4, with small and landlocked countries affected the most.
Groups of landlocked countries, such as Switzerland and Austria, that all need to increase (decrease) their sizes pose significant difficulties. While being adjacent to the sea helps, it does not always suffice to reach the target area: autobahn and motorway giants such as Germany and the UK, need to further increase their areas but are eventually blocked by other countries. This example suggests that for cartograms with large area changes it is necessary to allow additional distortions, e.g., by allowing  vertex movement in order to decrease cartographic error; we have not considered such an approach yet.

In summary, the examples illustrate the utility of circular-arc cartograms. 
they are readable (countries are where they should be), they are recognizable (the adjacencies between neighbors are preserved), they have low complexity, and they yield visually appealing country shapes that immediately communicate whether regions increase or decrease. Moreover, as the US presidential election example on Fig.~\ref{fig:usa} shows, circular-arc cartograms make it possible to spot patterns in the data when another parameter is encoded with color.
On the other hand, with our current heuristic we cannot guarantee low cartographic error if drastic area changes that are required by the data.
This is not necessarily a downside of circular-arc cartograms themselves, but rather due to our heuristic used to compute these examples.
Other more abstract cartogram types, e.g., rectangular cartograms~\cite{ks-rc-07,raisz} or circle cartograms~\cite{Dorling}, typically achieve very low area errors but come at the cost of lower recognizability as they often change adjacency relationships between neighboring countries. Finally, all of the examples in this paper are of weak circular-arc cartograms, and strong circular-arcs might be preferable.
In the next section we briefly discuss several possible approaches to decrease cartographic errors for circular-arc cartograms.

\section{Conclusions and Future Work}\label{sec:conclusion}
In this paper we introduced circular-arc cartograms as a new model for value-by-area diagrams. We showed that the \CAC problem is NP-hard and presented a heuristic algorithm to produce valid circular-arc cartograms with fairly low cartographic error. 
The results from our implementation indicate that circular-arc cartograms are readable, recognizable, have low complexity and are generally visually appealing. While for many countries in our %
examples the cartographic error is low, this cannot be guaranteed. 
There are several natural directions for future algorithmic and experimental work on circular-arc cartograms.

First, we note that the potential area change for each edge depends on its input length. 
Thus the fewer and longer the edges of a face boundary are, the larger is the range of realizable areas of the face. 
While we assumed that a fixed (simplified) subdivision is given as input, we can also allow further simplification on demand, i.e., the larger the required area changes the more polygon vertices are discarded in order to create longer and fewer edges. Such an approach preserves well the shape and complexity of regions with low area change, whereas at the same time strongly distorted regions with large area change become more strongly simplified. We could also allow introducing gaps with the shape of biconvex lenses between two neighboring countries that both need to decrease their areas; this idea corresponds to splitting these edges into two, each of which can then be bent inwards. 

Second, while it is generally undesirable to displace regions, it seems often possible to obtain lower cartographic errors by displacing just a few boundary vertices. It is natural to consider the trade-off between minimizing the overall cartographic error and minimizing overall vertex movement.

Third, we need to further study the effect of weak and strong circular-arc cartograms on error rates and perception. Recall that in the strong version all edges of a deflated country point inwards, while in the weak model (used in this paper) we allow some edges to point out. 

Fourth, one of the appealing features of circular-arc cartograms is the easily-interpretable cloud-like shape of countries that have increased area and snowflake-shape of countries with decreased area. Generalizing circular arcs to other types of smooth curves, e.g., cubic splines, may result in visually similar cartograms which allow for more flexibility and better accuracy. 

Ultimately, there is a need for a formal evaluation of the utilities of cartograms in general and of circular-arc cartograms in particular. It would natural to expect that readability, recognizability, faithfulness and complexity would vary in importance, depending on the given task. Determining the ``best'' cartograms would be a difficult but worthwhile goal.

\medskip
\noindent\textbf{Acknowledgments.} We thank Wouter Meulemans for providing us with schematized map instances. Research supported in part
by NSF grants CCF-1115971, DEB 1053573 and by the \emph{Concept for the Future} of KIT within the
framework of the German Excellence Initiative.

\bibliographystyle{abbrv+}
\bibliography{abbrv,masterReferences,missingRefs} 

\appendix
\section{Further examples}

\begin{figure}[h]
	\centering
	\includegraphics[width=\columnwidth]{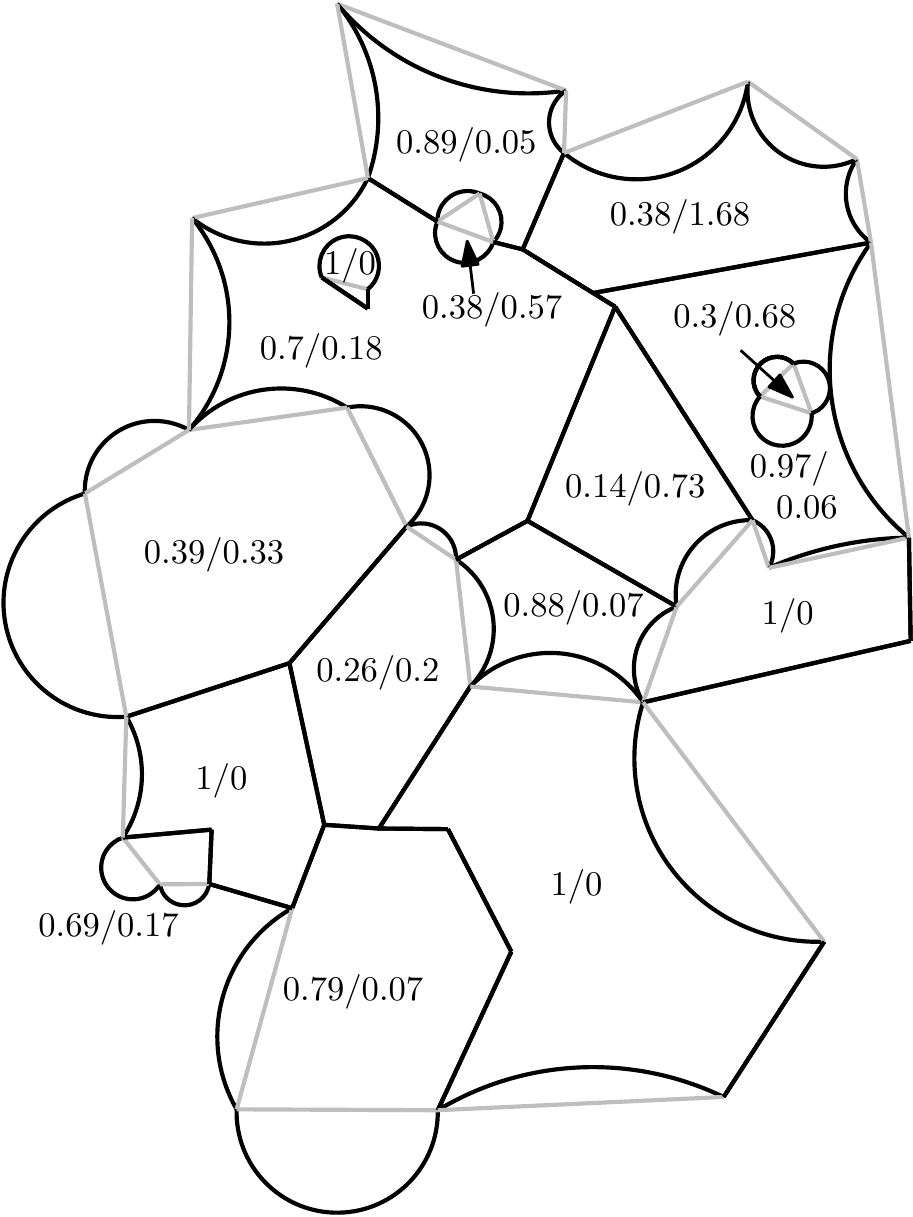}
	\caption{Population cartogram for the states of Germany.
}
	\label{fig:germany_pop}
\end{figure}

Figures~\ref{fig:germany_pop}--\ref{fig:germany_rail_fine} show cartograms of the states of Germany. We used three different data sets: population data\footnote{2011 population data from
\url{http://www.statistik-portal.de/Statistik-Portal/de_jb01_jahrtab1.asp}.}, number of craft enterprises\footnote{2009 craft enterprise data from \url{http://www.statistik-portal.de/Statistik-Portal/de_jb19_jahrtab1.asp}}, and railroad kilometers\footnote{2010 railroad data from \url{https://www.destatis.de/DE/ZahlenFakten/Wirtschaftsbereiche/TransportVerkehr/UnternehmenInfrastrukturFahrzeugbestand/Tabellen/Schieneninfrastruktur.html}}. The underlying input map was simplified by hand in a way that the most characteristic shapes are preserved and yet only relatively few edges per polygon remain. Unlike in the previous examples, we did not restrict the edge slopes.

The average success rate in Figure~\ref{fig:germany_pop} is 0.67 and the average cartographic error is 0.3. While several states are error-free or perform fairly well, there are notable examples of densely populated states like Berlin that need to grow a lot further and states like Mecklenburg-Vorpommern in Northern Germany that are sparsely populated and need to shrink a lot more. In both cases low success rates and high errors are observed; similarly the landlocked states perform worse than those on the boundary.

\begin{figure}[h]
	\centering
	\includegraphics[width=\columnwidth]{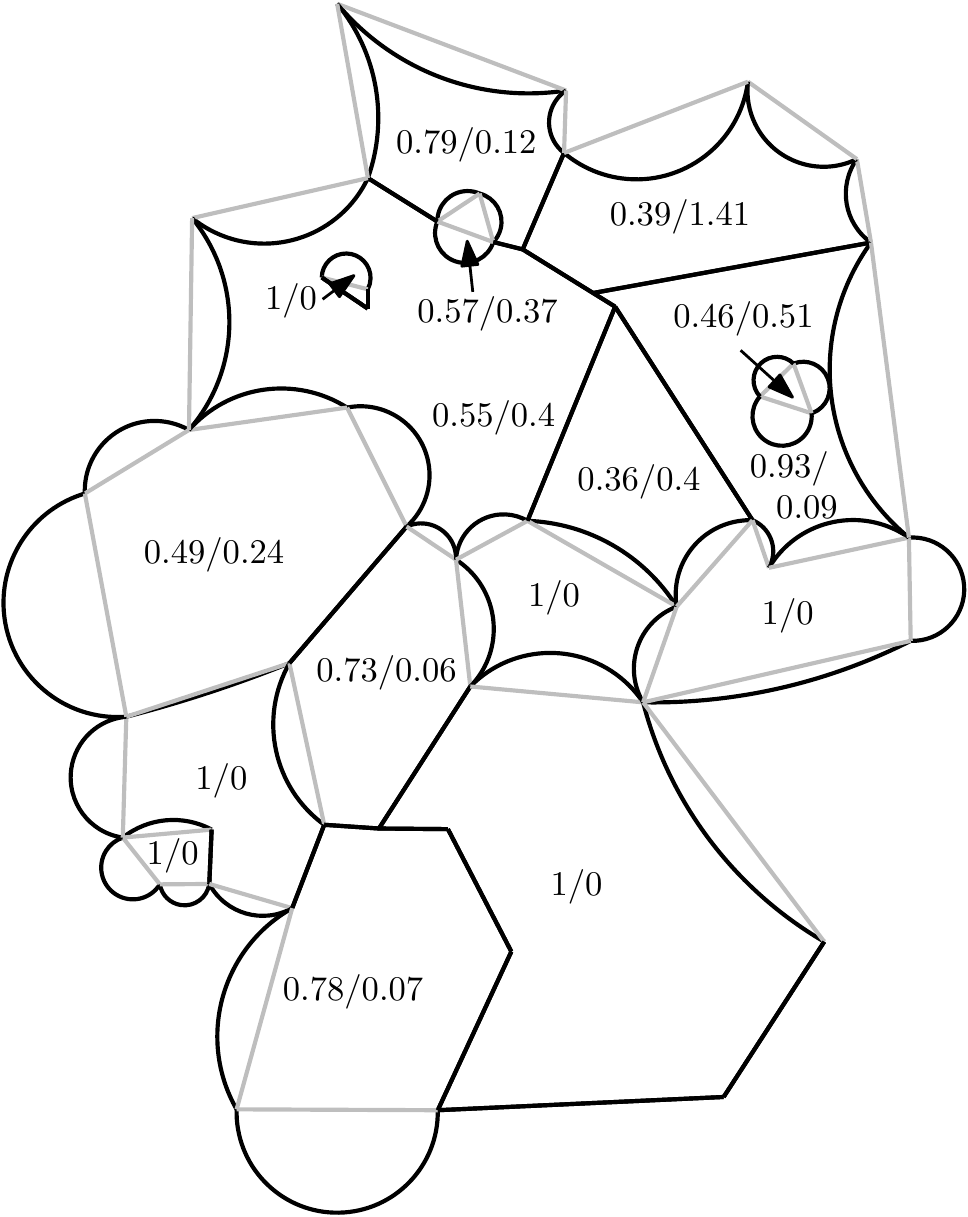}
	\caption{Cartogram for the number of crafts enterprises in the states of Germany.
}
	\label{fig:germany_craft}
\end{figure}

Figure~\ref{fig:germany_craft} shows interesting statistical data in the sense that there is a clear difference between the states in the North (except the city states Bremen, Hamburg and Berlin) with fewer enterprises and the states in the South with more enterprises. The average success rate in this example is 0.75 and the average error rate is 0.23. Thus we see slightly better results than for the population data. As expected, the problematic states are again those that are very densely populated and the sparse state of Mecklenburg-Vorpommern. Although the accuracy is not yet fully satisfying, the overall trend in the data is nicely conveyed. The southern states as well as the cities are all cloud-shaped having many craft enterprises and the northern states are all snowflake-shaped meaning fewer craft enterprises.

\begin{figure}[h]
	\centering
	\includegraphics[width=\columnwidth]{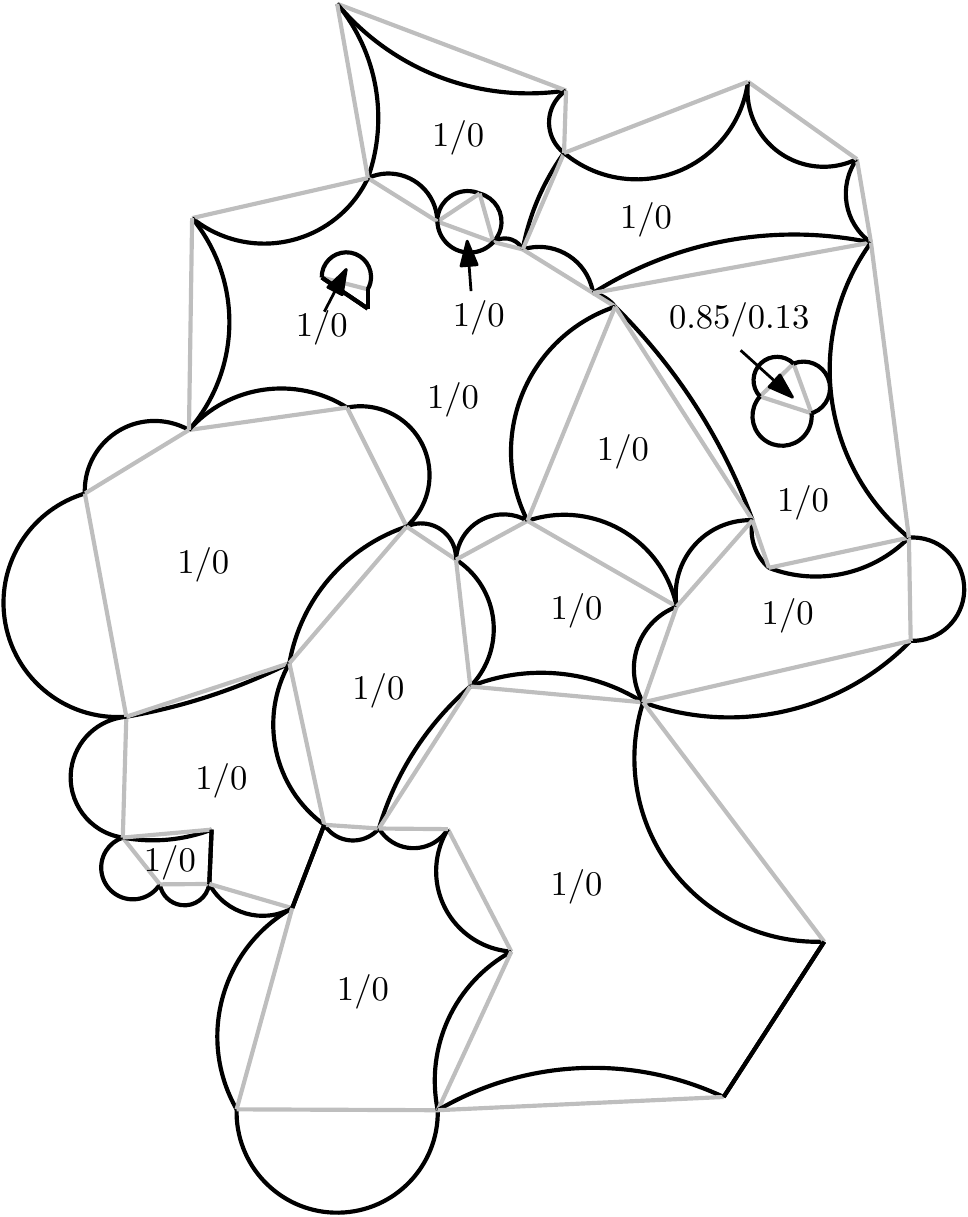}
	\caption{Cartogram for the railroad kilometers in the states of Germany.
}
	\label{fig:germany_rail}
\end{figure}

\begin{figure}[h]
	\centering
	\includegraphics[width=\columnwidth]{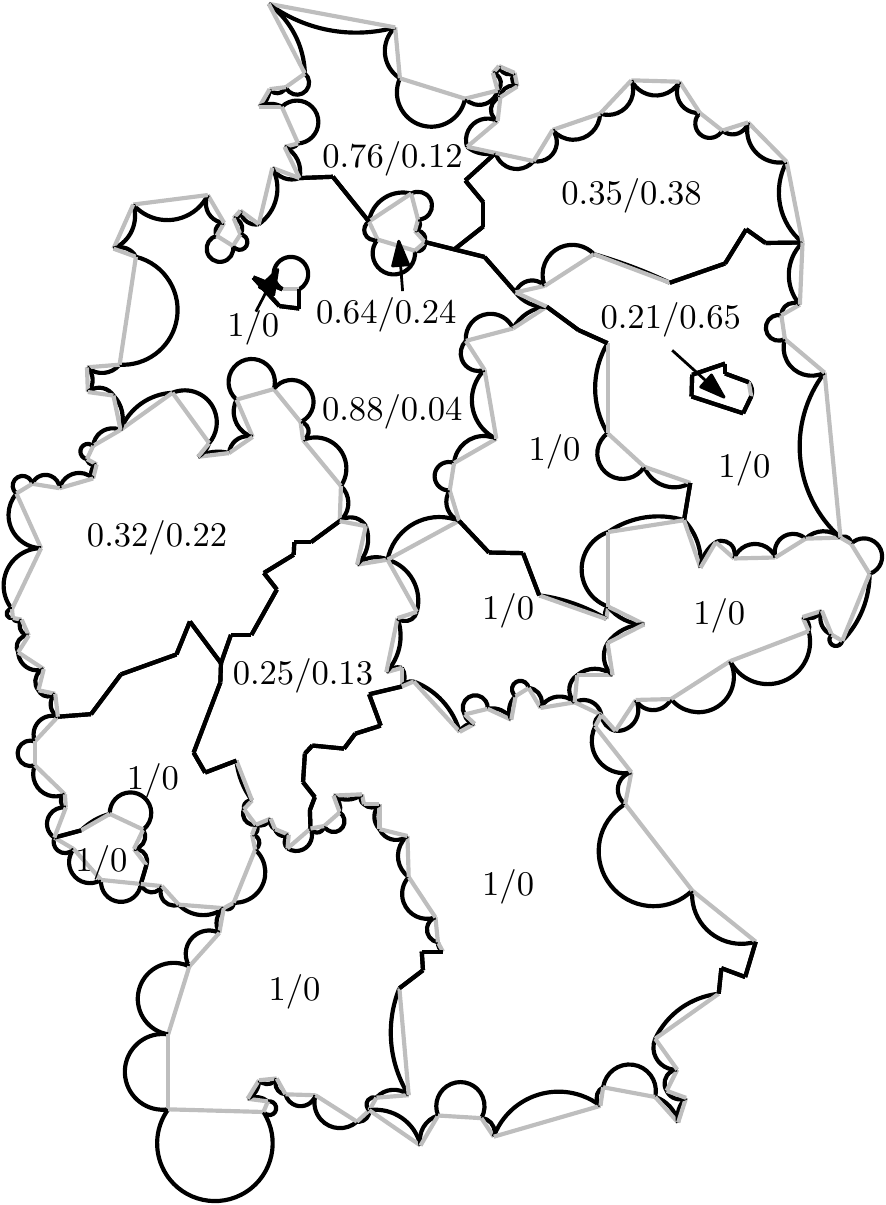}
	\caption{Cartogram for the railroad kilometers in the states of Germany using a more detailed input map.
}
	\label{fig:germany_rail_fine}
\end{figure}

Finally, Figures~\ref{fig:germany_rail} and~\ref{fig:germany_rail_fine} show two cartograms for the same data set of railroad kilometers per state. Figure~\ref{fig:germany_rail} uses the same input map as the previous two examples and Figure~\ref{fig:germany_rail_fine} uses a more detailed input map with a lot more and a lot shorter edges defining the state polygons. In Figure~\ref{fig:germany_rail} we see only a single state with an area error, namely Berlin which has a very extensive railway network compared to its area. Consequently the average success rate is more than 0.99 and the average area error is less than 0.01. Even all the landlocked states achieve their target areas completely. Since this example performs so well, it is interesting to study the effects of increasing the shape complexity of the input map by adding back in more details. 
The map in Figure~\ref{fig:germany_rail_fine} uses more than four times as many edges as Figure~\ref{fig:germany_rail}. While the performance of this cartogram is still reasonably good with an average success rate of 0.78 and an average error of 0.11, it gets clear in the direct comparison that the capacity of bending the shorter edges is not sufficient to reach all target areas, both for growing and for shrinking states. On the other hand the similarity to the true geographic shapes is higher in this cartogram. Nonetheless, the stylized appearance of Figure~\ref{fig:germany_rail} with fewer and longer arcs makes this cartogram more appealing for the purpose of depicting statistical data on a very abstract map. The comparison of the two cartograms and their performance shows that strong simplification of the input shapes is generally advisable, both for better aesthetics and to achieve lower cartographic errors.

\end{document}